\documentclass{llncs}%[envcountsame]

\usepackage[english]{babel}
\usepackage[T1]{fontenc}
\usepackage{lmodern}
\usepackage[utf8]{inputenc}
\usepackage[final]{microtype}

\usepackage{amsmath,amssymb,xcolor,graphicx,wrapfig,paralist}
\usepackage{xparse}
\usepackage{mathtools}
\usepackage{environ}
\usepackage{etoolbox}
\usepackage{tabu,multirow}
\usepackage{csquotes}

\newcolumntype{L}{X}
\newcolumntype{R}{>{\raggedleft\arraybackslash}X}
\newcolumntype{C}{>{\centering\arraybackslash}X}

\RequirePackage{marvosym}

\DeclarePairedDelimiter{\delimabs}{\lvert}{\rvert}
\DeclarePairedDelimiter{\delimnorm}{\lVert}{\rVert}
\DeclarePairedDelimiter{\delimpospart}{\lgroup}{\rgroup^+}

\DeclarePairedDelimiterX{\deliminner}[2]{\lange}{\rangle}{#1, #2}
\DeclarePairedDelimiter{\delimcardinality}{\lvert}{\rvert}
\DeclarePairedDelimiter{\delimset}{\lbrace}{\rbrace}
\DeclarePairedDelimiter{\delimtuple}{(}{)}
\DeclarePairedDelimiter{\delimlistt}{[}{]}
\DeclarePairedDelimiter{\delimfun}{(}{)}

\NewDocumentCommand{\abs}{sm}{\IfBooleanTF{#1}{\delimabs{#2}}{\delimabs*{#2}}}
\NewDocumentCommand{\norm}{sm}{\IfBooleanTF{#1}{\delimnorm{#2}}{\delimnorm*{#2}}}
\NewDocumentCommand{\pospart}{sm}{\IfBooleanTF{#1}{\delimpospart{#2}}{\delimpospart*{#2}}}
\NewDocumentCommand{\negpart}{sm}{\IfBooleanTF{#1}{\delimnetpart{#2}}{\delimnetpart*{#2}}}
\NewDocumentCommand{\inner}{sm}{\IfBooleanTF{#1}{\deliminner{#2}}{\deliminner*{#2}}}
\NewDocumentCommand{\cardinality}{sm}{\IfBooleanTF{#1}{\delimcardinality{#2}}{\delimcardinality*{#2}}}
\NewDocumentCommand{\set}{sm}{\IfBooleanTF{#1}{\delimset{#2}}{\delimset*{#2}}}
\NewDocumentCommand{\tuple}{sm}{\IfBooleanTF{#1}{\delimtuple{#2}}{\delimtuple*{#2}}}
\NewDocumentCommand{\closure}{sm}{\IfBooleanTF{#1}{\delimclosure{#2}}{\delimclosure*{#2}}}
\NewDocumentCommand{\listt}{sm}{\IfBooleanTF{#1}{\delimlistt{#2}}{\delimlistt*{#2}}}
\NewDocumentCommand{\fun}{smm}{\IfBooleanTF{#1}{{#2}\delimfun{#3}}{{#2}\delimfun*{#3}}}
\NewDocumentCommand{\funMacro}{smm}{\IfNoValueTF{#3}{#1}{\fun{#2}{#3}}}

\DeclareMathOperator{\ExistsOp}{\exists}
\DeclareMathOperator{\ForallOp}{\forall}

\NewDocumentCommand{\Exists}{gg}{\IfNoValueTF{#1}{\ExistsOp}{\ExistsOp #1. \, #2}}
\NewDocumentCommand{\Forall}{gg}{\IfNoValueTF{#1}{\ForallOp}{\ForallOp #1. \, #2}}

\newcommand{\unionSym}{\cup}
\newcommand{\unionBin}{\mathbin{\unionSym}}

\newcommand{\intersectionSym}{\cap}
\newcommand{\intersectionBin}{\mathbin{\intersectionSym}}
\newcommand{\UnionSym}{\bigcup}

\newcommand{\union}{\unionBin}

\newcommand{\intersection}{\intersectionBin}
\newcommand{\Union}{\UnionSym}

\newcommand{\Naturals}{\mathbb{N}}

\newcommand{\Reals}{\mathbb{R}}

\newcommand{\Distributions}{\mathcal{D}}

\NewDocumentCommand{\convto}{G{}}{\xrightarrow{#1}}
\NewDocumentCommand{\weakto}{G{}}{\xrightharpoonup{#1}}
\NewDocumentCommand{\weakstarto}{G{}}{\xrightharpoonup[*]{#1}}

\DeclareMathOperator{\supp}{supp}

\DeclareMathOperator*{\argmax}{arg\, max}

\NewDocumentCommand{\distributions}{d()}{\funMacro{\mathcal{D}}{#1}}

\newcommand{\gain}{g} %value
\newcommand{\bias}{b}

\newcommand{\rmax}{r_{\max}}

\newcommand{\Mc}{\mathsf{M}}
\newcommand{\Mdp}{\mathcal{M}}

\newcommand{\states}{S}

\newcommand{\initstate}{s_\textrm{init}}

\newcommand{\trans}{\Delta}
\newcommand{\av}{\mathsf{Av}}
\newcommand{\rew}{r}

\newcommand{\path}{\rho}
\newcommand{\fpath}{w}

\newcommand{\straa}{\pi}
\newcommand{\straas}{\Pi}

\NewDocumentCommand{\actions}{d()}{{\IfNoValueTF{#1}{\mathit{Act}}{\fun{\mathit{Act}}{#1}}}}

\newcommand{\mec}{\mathsf{MEC}}
\newcommand{\scc}{\mathsf{SCC}}
\newcommand{\bscc}{\mathsf{BSCC}}

\newcommand{\attractor}{\mathsf{prob1}}
\newcommand{\pr}{\mathbb P}
\newcommand{\expected}{\mathbb{E}}
\newcommand{\expsucc}{\expected_\trans}
\usepackage{algorithmicx,algorithm}
\usepackage[noend]{algpseudocode}
\usepackage{booktabs}
\usepackage{proof}
\usepackage{tikz}
\usepackage{graphics}
\usepackage{stmaryrd}
\usepackage[caption=false]{subfig}

\RequirePackage{url}

\usetikzlibrary{shapes,positioning,arrows,calc,automata,matrix,fit}
\tikzstyle{state}+=[minimum size = 6mm, inner sep=0,outer sep=1]
\tikzset{->,>=stealth'}

\newcommand{\secspace}{\vspace*{-0.75em}}
\newcommand{\subsecspace}{\vspace*{-0.5em}}
\newcommand{\subsubspace}{\vspace*{-1em}}
\renewcommand{\secspace}{}
\renewcommand{\subsecspace}{}
\renewcommand{\subsubspace}{}

\title{Efficient Strategy Iteration for Mean Payoff in Markov Decision Processes}%FINAL\thanks{This work is partially supported by the German Research Foundation (DFG) project \enquote{Verified Model Checkers} and the Czech Science Foundation grant No.~\mbox{15-17564S}.}}
\author{Jan~K{\v r}et\'insk\'y \and Tobias~Meggendorfer}
\institute{Technical University of Munich}

\begin{document}
\maketitle %
\begin{abstract}
	Markov decision processes (MDPs) are standard models for probabilistic systems with non-deterministic behaviours.
	Mean payoff (or long-run average reward) provides a mathematically elegant formalism to express performance related properties.
	Strategy iteration is one of the solution techniques applicable in this context.
	While in many other contexts it is the technique of choice due to advantages over e.g.\ value iteration, such as precision or possibility of domain-knowledge-aware initialization, it is rarely used for MDPs, since there it scales worse than value iteration.
	We provide several techniques that speed up strategy iteration by orders of magnitude for many MDPs, eliminating the performance disadvantage while preserving all its advantages.
\end{abstract} %
\section{Introduction}

\emph{Markov decision processes (MDPs)}~\cite{Howard,FV97,Puterman} are a standard model for analysis of systems featuring both probabilistic and non-deterministic behaviour.
They have found rich applications, ranging from communication protocols to biological systems and robotics.
A classical objective to be optimized in MDPs is \emph{mean payoff} (or \emph{long-run average reward}).
It captures the reward we can achieve on average per step when simulating the MDP.
Technically, one considers partial averages (average over the first $n$ steps) and let the time $n$ go to infinity.
This objective can be used to describe performance properties of systems, for example, average throughput, frequency of errors, average energy consumption, etc.

\emph{Strategy} (or \emph{policy}) \emph{iteration} (or \emph{improvement}) (SI) is a dynamic-programming technique applicable in many settings, including optimization of mean payoff in MDPs~\cite{Howard,Puterman}, but also mean payoff games \cite{DBLP:journals/dam/BjorklundV07,DBLP:journals/ijfcs/BrimC12}, parity games \cite{DBLP:conf/cav/VogeJ00,DBLP:conf/csl/Schewe08,DBLP:journals/corr/abs-0806-2923,fearnley17}, simple stochastic games \cite{DBLP:conf/dimacs/Condon90}, concurrent reachability games \cite{DBLP:journals/mst/HansenIM14}, or stochastic parity games \cite{DBLP:conf/vmcai/HahnST017}.
The main principle of the technique is to start with an arbitrary strategy (or policy or controller of the system) and iteratively improve it locally in a greedy fashion until no more improvements can be done.
The resulting strategy is guaranteed to be optimal.

SI has several advantages compared to other techniques used in these contexts.
Most interestingly, domain knowledge or heuristics can be used to \emph{initialize} with a reasonable strategy, thus speeding up the computation to a fraction of the usual analysis time.
Further, SI is conceptually simple as it boils down to a search through a \emph{finite space} of memoryless deterministic strategies, yielding arguments for correctness and termination of the algorithm.

More specifically, in the context of MDPs, it has advantages over the other two standard techniques.
Firstly, compared to \emph{linear programming} (LP), SI \emph{scales} much better.
LP provides a rich framework, which is able to encode many optimization problems on MDPs and in particular mean payoff.
However, although the linear programs are typically of polynomial size and can be also solved in polynomial time, such procedures are not very useful in practice.
For larger systems the solvers often time out or run out of memory already during the construction of the linear program.
Furthermore, SI ensures that the current lower bounds on the mean payoff is \emph{monotonically improving}.
Consequently, the iteration can be stopped at any point, yielding a strategy at least as good as all the previous iterations.

Secondly, compared to \emph{value iteration} (VI), SI provides a \emph{precise solution}, whereas VI is only optimal in the limit and the number of iterations before the numbers can be rounded in order to obtain a precise solution is very high \cite{krish-survey}.
Furthermore, stopping criteria for VI are limited to special cases or are very inefficient.
Consequently, VI is practically used to produce results that may be erroneous even for simple, realistic examples in verification, see e.g.~\cite{haddad2014reachability}.

%(1) \emph{Linear programming} provides a rich framework, which is able to encode many optimization problems on MDPs and in particular mean payoff.
%Although the linear programs are typically of polynomial size and thus can be also solved in polynomial time, such procedures are not very useful for larger systems since the solvers time out or run out of memory even during the construction of the linear program.
%(2) \emph{Value iteration} is a practically useful approach despite its exponential-time complexity.
%This iterative method converges in practical cases very fast.
%One of the main disadvantages is the general lack of good (or even any) stopping criteria.
%Consequently, this leads to erroneous results of this class of methods even for realistic and very simple examples in verification.
%(3) \emph{Strategy} (or \emph{policy}) \emph{iteration} (or \emph{improvement}) is another dynamic-programming approach with exponential-time complexity.

On the other hand, the main disadvantage of SI, in particular for mean payoff, is its \emph{scalability}.
Although SI scales better than LP, it is only rarely the case that SI is faster than VI.
Firstly, in the worst case, we have to examine \emph{exponentially} many strategies \cite{DBLP:conf/icalp/Fearnley10}, in contrast to the discounted case, which is polynomial (for a fixed discount factor) \cite{DBLP:journals/mor/Ye11a} even for games \cite{DBLP:journals/jacm/HansenMZ13}.
However, note that even for parity games it was for long not known \cite{DBLP:conf/lics/Friedmann09} whether all SI algorithms exhibit this property since the number of improvements is only rarely high in practice.
Secondly, and more importantly, the \emph{evaluation} of each strategy necessary for the greedy improvement takes enormous time since large systems of linear equations have to be solved.
Consequently, VI typically is much faster than SI to obtain a similar precision, although it may also need an exponential number of updates.

This scalability limitation is even more pronounced by the following contrast.
On the one hand, mean payoff games, parity games, and simple stochastic games are not known to be solvable in polynomial time, hence the exponential-time SI is an acceptable technique for these models.
On the other hand, for problems on MDPs that are solvable in polynomial time, such as mean payoff, the exponential-time SI becomes less appealing.
In summary, we can only afford to utilize the mentioned advantages of SI for MDPs if we make SI perform well in practice.

This paper suggest several heuristics and opens new directions to increase performance of SI for MDPs, in particular in the setting of mean payoff.
Our contribution is the following:
\begin{itemize}
	\item
	We present several techniques to significantly speed up SI in many cases, most importantly the evaluation of the current strategy.
	The first set of techniques (in Section~\ref{sec:topological_opt}) is based on maximal end component decomposition of the MDP and strongly connected component decomposition of the Markov chain induced by the MDP and the currently considered strategy.
	The second class (in Section \ref{sec:heuristic_si}) is based on approximative techniques to compute mean payoff in these Markov chains.
	Both variants reduce the time taken by the strategy evaluation.
	Finally, we combine the two approaches in a non-trivial way in Section~\ref{sec:mec_heuristics}, giving rise to synergic optimizations and opening the door for approximation techniques.
	\item
	We provide experimental evaluation of the proposed techniques and compare to the approaches from literature.
	We show experimental evidence that our techniques are speeding up SI by orders of magnitude and make its performance (i)~on par with VI, the prevalent technique which, however, only provides approximate solutions, and (ii)~incomparably more scalable than the precise technique of LP.
\end{itemize}
\paragraph{Further related work}
Strategy iteration for MDPs has been extensively studied~\cite{Howard,Puterman,Fearnley}.
Performance of SI for MDPs has been mainly improved in the discounted total reward case by, e.g., approximate evaluation of the strategy using iterative methods of linear algebra~\cite{Shlakhter}, model reduction by adaptive state-space aggregation~\cite{milan} or close-to-optimal initialization~\cite{Solis}; for an overview see \cite{Bertsekas}.
The treatment of the undiscounted case has focused on unichain MDPs \cite{unichains,Puterman}.
Apart from solving the MDPs modelling probabilistic systems, the technique has found its applications in other domains, too, for example program analysis \cite{Helmut}. %
\secspace %
\section{Preliminaries} \label{sec:prelim}
\secspace %

In this section, we introduce some central notions.
Furthermore, for the reader's convenience, Appendix~\ref{sec:linear_algebra} recalls some technical notions from linear algebra. %FINAL

A \emph{probability distribution} on a finite set $X$ is a mapping $\rho: X \to [0,1]$, such that $\sum_{x\in X} \rho(x) = 1$.
Its \emph{support} is denoted by $\supp(\rho) = \set{x \in X \mid \rho(x) > 0}$.
$\Distributions(X)$ denotes the set of all probability distributions on $X$.
\begin{definition}
	A \emph{Markov chain (MC)} is a tuple $\Mc = (\states, \initstate, \trans, \rew)$, where $\states$ is a finite set of \emph{states}, $\initstate \in \states$ is the \emph{initial} state, $\trans: \states \to \distributions(\states)$ is a \emph{transition function} that for each state $s$ yields a probability distribution over successor states and $\rew : \states \to \Reals^{\geq 0}$ is a \emph{reward function}, assigning rewards to states.
\end{definition}
\begin{definition}
	A \emph{Markov decision process (MDP)} is a tuple of the form $\Mdp = (\states, \initstate, \actions, \av, \trans, \rew)$, where $\states$ is a finite set of \emph{states}, $\initstate \in \states$ is the \emph{initial} state, $\actions$ is a finite set of \emph{actions}, $\av: \states \to 2^{\actions}$ assigns to every state a set of \emph{available} actions, $\trans: \states \times \actions \to \distributions(\states)$ is a \emph{transition function} that for each state $s$ and action $a \in \av(s)$ yields a probability distribution over successor states and $\rew : \states \times \actions \to \Reals^{\geq 0}$ is a \emph{reward function}, assigning rewards to state-action pairs.

	Furthermore, we assume w.l.o.g. that actions are unique for each state, i.e.\ $\av(s) \intersection \av(s') = \emptyset$ for $s \neq s'$.\footnote{The usual procedure of achieving this in general is to replace $\actions$ by $\states \times \actions$ and adapting $\av$, $\trans$, and $\rew$ appropriately.
	For the sake of readability, we omit this restriction when drawing examples.}
\end{definition}
For ease of notation, we overload functions mapping to distributions $f: Y \to \Distributions(X)$ by $f: Y \times X \to [0, 1]$, where $f(y, x) := f(y)(x)$.
For example, instead of $\trans(s)(s')$ and $\trans(s, a)(s')$ we write $\trans(s, s')$ and $\trans(s, a, s')$, respectively.
Further, given some MC $\Mc$, a function $f : \states \to \Reals$ and set of states $C \subseteq \states$, we define $\expsucc^C(f, s) := \sum_{s' \in C} \trans(s, s') f(s')$, i.e.\ the weighted sum of $f$ over all the successors of $s$ in $C$.
Analogously, for some MDP $\Mdp$, we set $\expsucc^C(f, s, a) := \sum_{s' \in C} \trans(s, a, s') f(s')$.
For $C = \states$, we omit the superscript, i.e.\ $\expsucc(f, s) := \expsucc^\states(f, s)$ and $\expsucc(f, s, a) := \expsucc^\states(f, s, a)$.

An \emph{infinite path} $\path$ in a Markov chain is an infinite sequence $\path = s_0 s_1 \dots \in \states^\omega$, such that for every $i \in \Naturals$ we have that $\trans(s_i, s_{i+1}) > 0$.
A \emph{finite path} $\fpath = s_0 s_1 \dots s_n \in \states^*$ is a finite prefix of an infinite path.
Similarly, an \emph{infinite path} in an MDP is some infinite sequence $\path = s_0 a_0 s_1 a_1 \dots \in (\states \times \actions)^\omega$, such that for every $i \in \Naturals$, $a_i\in \av(s_i)$ and $\trans(s_i,a_i, s_{i+1}) > 0$.
\emph{Finite path}s are defined analogously as elements of $(\states \times \actions)^* \times \states$.

A Markov chain together with a state $s$ induces a unique probability distribution $\pr_s$ over measurable sets of infinite paths \cite[Ch.~10]{BaierBook}.
For some $C \subseteq \states$, we write $\Diamond C$ to denote the set of all paths which eventually reach $C$, i.e. $\Diamond C = \set{\path = s_0 s_1 \dots \mid \exists i \in \Naturals.~s_i \in C}$, which is measurable \cite[Sec.~10.1.1]{BaierBook}.

A \emph{strategy} on an MDP is a function $\straa: (\states \times \actions)^*\times S \to \distributions(\actions)$, which given a finite path $\fpath = s_0 a_0 s_1 a_1 \dots s_n$ yields a probability distribution $\straa(\fpath) \in \distributions(\av(s_n))$ on the actions to be taken next.
We call a strategy \emph{memoryless randomized} (or \emph{stationary}) if it is of the form $\straa: \states \to \distributions(\actions)$, and \emph{memoryless deterministic} (or \emph{positional}) if it is of the form $\straa: \states \to \actions$.
We denote the set of all strategies of an MDP by $\straas$, and the set of all memoryless deterministic strategies as $\straas^{\mathsf{MD}}$.
Note that $\straas^{\mathsf{MD}}$ is finite, since at each state there exist only finitely many actions to choose from.
Fixing any positional strategy $\straa$ induces a Markov chain where $\trans(s, s') = \sum_{s \in \av(s)} \straa(s, a) \cdot \trans(s, a, s')$ and $\rew(s) = \sum_{a \in \av(s)} \straa(s, a) \cdot \rew(s, a)$.

Fixing a strategy $\straa$ and an initial state $s$ on an MDP $\Mdp$ also gives a unique measure $\pr^\straa_s$ over infinite paths~\cite[Sec.~2.1.6]{Puterman}.
The expected value of a random variable $F$ then is defined as $\expected^\straa_s[F] = \int F\ d\,\pr^\straa_s$.
\subsubspace %
\subsubsection{Strongly connected components and end components}
A non-empty set of states $C \subseteq \states$ in a Markov chain is \emph{strongly connected} if for every pair $s, s' \in C$ there is a path from $s$ to $s'$, possibly of length zero.
Such a set $C$ is a \emph{strongly connected component} (SCC) if it is inclusion maximal, i.e.\ there exists no strongly connected $C'$ with $C \subsetneq C'$.
Note that each state of an MC belongs to exactly one SCC\footnote{Some authors deliberately exclude so called \enquote{trivial} or \enquote{transient} SCCs, which are single states without a self-loop.}.
A SCC is called \emph{bottom strongly connected component} (BSCC) if additionally no path leads out of it, i.e.\ for $s \in C, s' \in S \setminus C$ we have $\trans(s, s') = 0$.
The set of SCCs and BSCCs in a MC $\Mc$ are denoted by $\scc(\Mc)$ and $\bscc(\Mc)$, respectively.

The concept of SCCs is generalized to MDPs by so called \emph{(maximal) end components}.
A pair $(T, A)$, where $\emptyset \neq T \subseteq S$ and $\emptyset \neq A \subseteq \Union_{s\in T} \av(s)$, is an \emph{end component} of an MDP $\Mdp$ if (i)~for all $s \in T, a \in A \intersection \av(s)$ we have $\supp(\trans(s,a)) \subseteq T$, and (ii)~for all $s, s' \in T$ there is a finite path $\fpath = s a_0 \dots a_n s' \in (T \times A)^* \times T$, i.e.\ the path stays inside $T$ and only uses actions in $A$.
Note that we assumed actions to be unique for each state.

Intuitively, an end component describes a set of states for which a particular strategy exists such that all possible paths remain inside these states.
An end component $(T, A)$ is a \emph{maximal end component (MEC)} if there is no other end component $(T', A')$ such that $T \subseteq T'$ and $A \subseteq A'$.
Given an MDP $\Mdp$, the set of its MECs is denoted by $\mec(\Mdp)$.

Finally, given an MDP $\Mdp$ let $(T, A) \in \mec(\Mdp)$ some MEC in it.
By picking some initial state $\initstate' \in T$, defining the straightforward restrictions of $\av$ and $\trans$ by $\av' : T \to 2^A$, $\av'(s) := \av(s) \intersection A$ and $\trans' : T \times A \to \Distributions(T)$, $\trans'(s, a) := \trans(s, a)$ one obtains the \emph{restricted MDP} $\Mdp' = (T, \initstate', A, \av', \trans')$.
\begin{remark} \label{rem:scc_and_mec_decomposition}
	For a Markov chain $\Mc$, the computation of $\scc(\Mc)$, $\bscc(\Mc)$ and a topological ordering of the SCCs can be achieved in linear time w.r.t. the number of states and transitions by, e.g., Tarjan's algorithm~\cite{tarjan1972depth}.
	Similarly, the MEC decomposition of an MDP can be computed in polynomial time \cite{CY95}.
	%For improved algorithms on general MDP and various special cases see~\cite{ChatterjeeH11,ChatterjeeH12,ChatterjeeH14}.
\end{remark}
\subsubspace %
\subsubsection{Long-run average reward} (also called \emph{mean payoff}) of a strategy $\straa$ intuitively describes the optimal reward we can expect on average per step when simulating the MDP according to $\straa$.
In the following, we will only consider the case of maximizing the average reward, but the presented methods easily can be transferred to the minimization case.

Formally, let $R_i$ be a random variable which for an infinite path $\path = s_0 a_0 s_1 a_1 \dots$ returns $R_i(\path) = \rew(s_i, a_i)$, i.e.\ the reward obtained at step $i \geq 0$.
Given a strategy $\straa$, the $n$-step (maximal) average reward then is defined as $\gain^\straa_n(s) = \expected^\straa_s (\frac1n\sum_{i=0}^{n-1} R_i)$.
The \emph{long-run average reward} (in this context also traditionally called \emph{gain} \cite{Puterman}) of the strategy $\straa$ is $\gain^\straa(s) = \liminf_{n\to\infty} \gain^\straa_n(s)$.\footnote{The $\liminf$ is used since the limit may not exist in general for an arbitrary strategy.}
Consequently, the \emph{long-run average reward} (or \emph{gain}) of a state $s$ is defined as
\begin{equation*}
	\gain^*(s) := \sup_{\straa \in \straas} \gain^\straa(s) = \sup_{\straa \in \straas} \liminf_{n\to\infty}\expected^\straa_s \left (\frac1n\sum_{i=0}^{n-1} R_i \right) .
\end{equation*}
%
%For finite MDP the maximal average reward $\gain^*(s)$ is attained by a memoryless deterministic strategy $\straa^*\in \straas^{\mathsf{MD}}$.
%Moreover, $\gain^*(s)$ is in fact the limit of the $n$-step average reward \cite{Puterman}.
For finite MDPs $\gain^*(s)$ in fact is attained by a memoryless deterministic strategy $\straa^*\in \straas^{\mathsf{MD}}$ and it further is the \emph{limit} of the $n$-step average reward~\cite{Puterman}.
Formally,
\begin{equation*}
	\gain^*(s) = \max_{\straa \in \straas^{\mathsf{MD}}} \gain^\straa(s) = \lim_{n\to\infty} \gain^{\straa^*}_n(s).
\end{equation*}
With this in mind, we now only consider memoryless deterministic strategies. %
\secspace %
\section{Strategy Iteration} \label{sec:si_basic}
\secspace %

One way of computing the optimal gain of an MDP (i.e. determining the optimal gain of each state) is \emph{strategy iteration} (or \emph{policy iteration} or \emph{strategy improvement}).
The general approach of strategy iteration is to (i)~fix a strategy, (ii)~evaluate it and (iii)~improve each choice greedily, repeating the process until no improvement is possible any more.
For an in depth theoretical expos\'e of strategy iteration for MDPs, we refer to e.g.\ \cite[Sec.~9.2]{Puterman}.
Here, we briefly recall the necessary definitions.

\subsecspace %
\subsubsection{Gain and bias}
As mentioned, the second step of strategy iteration requires to evaluate a given strategy.
By investigating the Markov chain $\Mc = (\states, \initstate, \trans, \rew)$ induced by the MDP $\Mdp$ together with a strategy $\straa \in \straas^{\mathsf{MD}}$, one can employ the following system of linear equations characterizing the gain $\gain$~\cite{Puterman}:
\begin{equation*}
	\begin{aligned}
		\gain(s) & = \sum_{s' \in \states} \trans(s, s') \cdot \gain(s') = \expsucc(\gain, s) \quad \forall s \in S, \\
		\bias(s) & = \sum_{s' \in \states} \trans(s, s') \cdot \bias(s') + \rew(s) - g(s) = \expsucc(\bias, s) + \rew(s) - g(s) \quad \forall s \in S.
	\end{aligned}
\end{equation*}
A solution $(\gain, \bias)$ to these \emph{gain/bias equations} yields the gain $\gain$ and the so called \emph{bias} $\bias$ of the induced Markov chain, which we also refer to as gain $\gain_\straa$ and bias $\bias_\straa$ of the corresponding strategy $\straa$.
Intuitively, the bias relates to the total expected deviation from the gain until the obtained rewards \enquote{stabilize} to the gain.
Note that the equations uniquely determine the gain but not the bias.
We refer the reader to \cite[Sec.~9.1.1, Sec.~9.2.1]{Puterman} for more detail but highlight the following result.
A unique solution can be obtained by adding the constraints $b(s_i) = 0$ for one arbitrary but fixed state $s_i$ in each BSCC~\cite[Condition~9.2.3]{Puterman}.
Note this condition requires to fix the bias of the \enquote{first} state in the BSCC to zero.
But, as the states can be numbered arbitrarily, any state of the BSCC is eligible.
This is also briefly touched upon in the corresponding chapter of \cite{Puterman}.
Unfortunately, this results in a non-square system matrix.

\begin{algorithm}[t]
	\caption{\textsc{SI}}
	\label{alg:si}
	\begin{algorithmic}[1]
		\Require MDP $\Mdp = (\states, \initstate, \actions, \av, \trans, \rew)$.
		\Ensure $(g^*, \straa^*)$, s.t.\ $g^*$ is the optimal gain of the MDP and is obtained by $\straa^*$.
		\State Set $n = 0$ and pick an arbitrary strategy $\straa_0 \in \straas^{\mathsf{MD}}$.
		\State Obtain $\gain_n$ and $\bias_n$ which satisfy the gain/bias equations. \label{alg:si:line:eval}
		\State Let \label{alg:si:line:gain_opt_actions} \Comment{Gain improvement}
		\begin{equation*}
			\av_{\gain_n}(s) = \argmax_{a \in \av(s)} \expsucc(\gain_n, s, a),
		\end{equation*}
		all actions maximizing the successor gains.
		\State Pick $\straa_{n+1} \in \straas^{\mathsf{MD}}$ s.t.\ $\straa_{n+1}(s) \in \av_{\gain_n}(s)$, setting $\straa_{n+1}(s) = \straa_n(s)$ if possible. \label{alg:si:line:gain_improve}
		\If{$\straa_{n+1} \neq \straa_n$} increment $n$ by 1 and go to Line~\ref{alg:si:line:eval}.
		\EndIf
		\State Pick $\straa_{n+1} \in \straas^{\mathsf{MD}}$ which satisfies \Comment{Bias improvement} \label{alg:si:line:bias_improve}
		\begin{equation*}
			\straa_{n+1}(s) \in \argmax_{a \in \av_{\gain_n}(s)} \rew(s, a) + \expsucc(\bias_n, s, a),
		\end{equation*}
		again setting $\straa_{n+1}(s) = \straa_n(s)$ if possible.
		\If{$\straa_{n+1} \neq \straa_n$} increment $n$ by 1 and go to Line~\ref{alg:si:line:eval}.
		\EndIf
		\State \Return $(\gain_{n+1}, \straa_{n+1})$.
	\end{algorithmic}
\end{algorithm}

With these results, the strategy iteration for the average reward objective on MDPs is defined in Algorithm~\ref{alg:si}\footnote{Note that the procedure found in~\cite[Sec.~9.2.1]{Puterman} differs from our Algorithm in Line~\ref{alg:si:line:bias_improve}.
That procedure indeed is erroneous, as the bias is improved over all available actions.
Optimizing the bias only over all actions which already optimize the gain is indeed vital to the idea.
The proofs provided in the corresponding chapter reflect this and actually prove the correctness of the algorithm as presented here.}.
Reasoning of \cite[Sec.~9.2.4]{Puterman} yields correctness.
\begin{theorem} \label{stm:si_correct}
	The strategy iteration presented in Algorithm~\ref{alg:si} terminates with a correct result for any input MDP.
\end{theorem}
It might seem unintuitive why the bias improvement in Line~\ref{alg:si:line:bias_improve} is necessary, since we are only interested in the gain after all.
Intuitively, when optimizing the bias the algorithm seeks to improve the expected \enquote{bonus} until eventually stabilizing without reducing the obtained gain.
This may lead to actually improving the overall gain, as illustrated in Appendix~\ref{sec:bias_improvement_necessary}. %FINAL

\subsubspace %
\subsubsection{Advantages and drawbacks of strategy iteration}
Compared to other methods for solving the average reward objective, e.g.\ value iteration \cite{cav,krish-survey}, strategy iteration offers some advantages:
\begin{enumerate}[(i)]
	\item A \emph{precise solution} can be obtained, compared to value iteration which is only optimal in the limit.
	\item The gain of the strategy is \emph{monotonically improving}, the iteration can be stopped at any point, yielding a strategy at least as good as the initial one.
	\item It therefore is easy to \emph{introduce knowledge} about the model or results of some pre-computation by initializing the algorithm with a sensible strategy.
	\item On some models, strategy iteration performs \emph{significantly faster} than value iteration, as outlined in Appendix~\ref{sec:strategy_faster_than_value_iteration}. %FINAL
	\item The algorithm searches through the \emph{finite space} of memoryless deterministic strategies, simplifying termination and correctness proofs.
	% \item Strategy iteration usually evaluates the reward function significantly less often than value iteration, making it more suitable for settings where \emph{evaluation of the reward function is costly}.
\end{enumerate}
But on the other hand, the naive implementation of strategy iteration as presented in Algorithm~\ref{alg:si} has several drawbacks:
\begin{enumerate}[(i)]
	\item In order to determine the precise gain by solving the gain/bias equations, one necessarily has to determine the bias, too.
	Therefore, the algorithm has to determine \emph{both gain and bias} in each step, while often only the gain is actually used for the improvement.
	\item For reasonably sized models the equation system becomes \emph{intractably large}.
	In the worst case, it contains $2 n^2 + n$ non-zero entries and even for standard models there often are significantly more than $n$ non-zero entries.
	\item Furthermore, the gain/bias equation system is under-determined, ruling out a lot of fast solution methods for linear equation systems.
	Uniqueness can be introduced by adding several rows, which results in the matrix being non-square, again ruling out a lot of solution methods.
	Experimental results suggest that this equation system furthermore has rather large condition numbers (see Appendix~\ref{sec:linear_algebra}) even for small, realistic models, leading to numerical instabilities\footnote{On crafted models with less than 10 states we observed numerical errors leading to non-convergence and condition numbers of up to $10^5$.}. %FINAL
	\item Lastly, the equation system is solved \emph{precisely} for every improvement step, which often is unnecessary.
	To arrive at a precise solution, we often only need to identify states in which the strategy is not optimal, compared to having a precise measure of how non-optimal they are.
\end{enumerate}
In the following two sections, we present approaches and ideas tackling each of the mentioned problems, arriving at procedures which perform orders of magnitude faster than the original approach. %
\secspace %
\section{Topological optimizations} \label{sec:topological_opt}
\secspace %

Our first set of optimizations is based on various topological arguments about both MDPs and MCs.
They are used to eliminate unnecessary redundancies in the equation systems and identify sub-problems which can be solved separately, eventually leading to small, full-rank equation systems.
Reduction in size and removal of redundancies naturally lead to significantly better condition numbers, which we also observed in our experiments.

\subsecspace %
\subsection{MEC decomposition} \label{sec:mec_dec}
\subsecspace %

We presented a variant of this method in our previous work \cite{cav} in the context of value iteration.
Due to space constraints we will only give a short overview of the idea.

The central idea is that all states in a MEC of some MDP have the same optimal gain~\cite[Sec.~9.5]{Puterman}\footnote{Restricting a general MDP to a MEC results in a \enquote{communicating} MDP.}.
Intuitively this is the case since any state in a particular MEC can reach every other state of the MEC almost surely.
For some MEC $M$ we define $g^*(M)$ to be this particular optimal value and call it the \emph{gain of the MEC}.
The optimal gain of the whole MDP then can be characterized by
\begin{equation*}
	\gain^*(s) = \max_{\straa \in \straas^{\mathsf{MD}}} \sum_{M \in \mec(\Mdp)} \pr^\straa_s[\Diamond\Box M] \cdot \gain^*(M)
\end{equation*}
where $\Diamond\Box M$ denotes the measurable set of paths that eventually remain within $M$.
This leads to a divide-and-conquer procedure for determining the gain of an MDP.
Conceptually, the algorithm first computes the MEC decomposition~\cite{CY95}, then for each MEC $M$ determines its gain $\gain^*(M)$ by strategy iteration and finally solves a reachability query on the \emph{weighted MEC quotient} $\Mdp^f$ by, e.g., strategy iteration or (interval) value iteration \cite{haddad2014reachability,atva}.

The weighted MEC quotient $\Mdp^f$ is a modification of the standard \emph{MEC quotient} of~\cite{DeAlfaro1997}, which for each MEC $M$ introduces an action leading from the collapsed MEC $M$ to a designated target sink $s_+$ with probability $f(M)$ (which is proportional to $\gain^*(M)$) and a non-target sink $s_-$ with the remaining probability.
With this construction, we can relate the maximal probability of reaching $s_+$ to the maximal gain in the original MDP.
For a formal definition, see Appendix~\ref{sec:weighted_quotient}. %FINAL

\begin{algorithm}[tb]
	\caption{\textsc{MEC-SI}}
	\label{alg:mec_dec_si}
	\begin{algorithmic}[1]
		\Require MDP $\Mdp = (\states, \initstate, \actions, \av, \trans, \rew)$.
		\Ensure The optimal gain $g^*$ of the MDP.
		\State $f \gets \emptyset$, $\rmax \gets \max_{s \in \states, a \in \av(s)} \rew(s, a)$.
		\For {$M_i = (T_i, A_i) \in \mec(\Mdp)$}
			\State Compute $\gain^*(M_i)$ of the MEC by applying Algorithm~\ref{alg:si} on the restricted MDP. \label{alg:mec_dec_si:line:mec_gain}
			\State Set $f(M_i) \gets \gain^*(M_i) / \rmax$.
		\EndFor
		\State Compute the weighted MEC quotient $\Mdp^{f}$.
		\State Compute $p \gets \pr^{\max}_{\Mdp^f}(\Diamond \set{s_+})$.
		\State \Return $\rmax \cdot p$
	\end{algorithmic}
\end{algorithm}

Using this idea, we define the first optimization of strategy iteration in Algorithm~\ref{alg:mec_dec_si}.
Its correctness follows from \cite[Theorem 2]{cav}.
Since we are only concerned with the average reward and each state in the restriction can reach any other (under some strategy), the initial state we pick for the restriction in Line~\ref{alg:mec_dec_si:line:mec_gain} is irrelevant.
Note that while the restricted MDP consists of a single MEC, an induced Markov chain may still contain an arbitrary number of (B)SCCs.

This algorithm already performs significantly better on a lot of models, as shown by our experimental evaluation in Section~\ref{sec:exper}.
But, as to be expected, on models with large MECs this algorithm still is rather slow compared to other approaches and may even add additional overhead when the whole model is a single MEC.
To this end, we will improve strategy iteration in general.
To combine these optimized variants with the ideas of Algorithm~\ref{alg:mec_dec_si}, one can simply apply them in Line~\ref{alg:mec_dec_si:line:mec_gain}.

\subsecspace %
\subsection{Using strongly connected components} \label{sec:sccs}
\subsecspace %

The underlying ideas of the previous approach are independent of the procedure used to determine $g^*(M)$.
Naturally, this optimization does not exploit any specific properties of strategy iteration to achieve the improvement.
In this section, we will therefore focus on improving the core principle of strategy iteration, namely the evaluation of a particular strategy $\straa$ on some MDP $\Mdp$.
As explained in Section~\ref{sec:si_basic}, this problem is equivalent to determining the gain and bias of some Markov chain $\Mc$.
Hence we fix such a Markov chain $\Mc$ throughout this section and present optimized methods for determining the required values precisely.

\subsubspace %
\subsubsection{BSCC compression}
In this approach, we try to eliminate superfluous redundancies in the equation system.
The basic idea is that all states in some BSCC have the same optimal gain.
Moreover, the same gain is achieved in the \emph{attractor} of $B$, i.e.\ all states from which almost all runs eventually end up in $B$.
\begin{definition}[Attractor]
	Let $\Mc$ be some Markov chain and $C \subseteq S$ some set of states in $\Mc$. The \emph{attractor} of $C$ is defined as
	\begin{equation*}
		\attractor(C) := \set{s \in S \mid \pr_s[\Diamond C] = 1},
	\end{equation*}
	i.e.\ the set of states which almost surely eventually reach $C$.
\end{definition}
\begin{lemma} \label{stm:gain_attractor_equal}
	Let $\Mc$ be a Markov chain and $B$ a BSCC.
	Then $\gain(s) = \gain(s')$ for all $s, s' \in \attractor(B)$.
\end{lemma}
\begin{proof}
	When interpreting the MC as a degenerate MDP with $\cardinality{\av(s)} = 1$ for all $s$, the gain of the MC coincides with the optimal gain of this MDP and each BSCC in the original MC is a MEC in the MDP.
	Using the reasoning from Section~\ref{sec:mec_dec} and \cite[Sec.~9.5]{Puterman}, we obtain that all states in $\attractor(B)$ have the same gain.
	\qed
\end{proof}
Therefore, instead of adding one gain variable per state to the equation system, we \enquote{compress} the gain of all states in the same BSCC (and its attractor) into one variable.
Formally, the reduced equation system is formulated as follows.

Let $\set{B_1, \dots, B_n} = \bscc(\Mc)$ be the BSCC decomposition of the Markov chain.
Further, define $A_i := \attractor(B_i)$ the attractors of each BSCC and $T := \Union_{i=1}^n A_i$ the set of all states which don't belong to any attractor.
The \emph{BSCC compressed gain/bias equations} then are defined as
\begin{equation}
	\label{eq:compressed_optimiality}
	\begin{aligned}
		\gain(s) & = \expsucc^{T'}(\gain, s) + {\sum}_{A_i} \expsucc^{A_i}(\gain_i, s) \quad \forall s \in T, \\
		\bias(s) & = \expsucc^{A_i}(\bias, s) + \rew(s) - \gain_i \quad \forall 1 \leq i \leq n, s \in A_i, \\
		\bias(s) & = \expsucc(\bias, s) + \rew(s) - \gain(s) \quad \forall s \in T, \\
		\bias(s_i) & = 0 \quad \text{for one arbitrary but fixed $s_i \in B_i$, } \forall 1 \leq i \leq n.
	\end{aligned}
\end{equation}
Applying the reasoning of Lemma~\ref{stm:gain_attractor_equal} immediately gives us correctness.
\begin{corollary} \label{stm:compression_correct}
	The values $g_1, \dots, g_n$, $g(s)$ and $b(s)$ are a solution to the equation system~(\ref{eq:compressed_optimiality}) if and only if
	\begin{equation*}
		g'(s) := \begin{dcases*}
			g_i & if $s \in A_i$, \\
			g(s) & otherwise.
		\end{dcases*}
	\end{equation*}
	and $\bias(s)$ are a solution to the gain/bias equations.
\end{corollary}
This equation system is significantly smaller for Markov chains which contain large BSCC-attractors.
Furthermore, observe that the resulting system matrix also is square.
We have $\cardinality{\bscc(\Mc)} + \cardinality{T}$ gain and $\cardinality{S}$ bias variables but also $\cardinality{T}$ gain and $\cardinality{S} + \cardinality{\bscc(\Mc)}$ bias equations.
Additionally, by virtue of Corollary~\ref{stm:compression_correct} and \cite[Condition~9.2.3]{Puterman}, the system has a unique solution.
Together, this allows the use of more efficient solvers.
Especially when combined with the previous MEC decomposition approach, significant speed-ups can be observed.
\subsubspace %
\subsubsection{SCC decomposition}
The second approach extends the BSCC compression idea by further decomposing the problem into numerous sub-problems.
The formal definition of the improved evaluation algorithm is given in Algorithm~\ref{alg:scc_eval}.
We explain the intuition below and prove correctness in Appendix~\ref{sec:proof_scc_eval}. %FINAL

\begin{algorithm}[t]
	\caption{\textsc{SCC-SI}}
	\label{alg:scc_eval}
	\begin{algorithmic}[1]
		\Require MC $\Mc = (\states, \initstate, \trans, \rew)$.
		\Ensure $(\gain, \bias)$, s.t.\ $\gain$ and $\bias$ are solutions to the gain/bias equations.
		\State Obtain $\bscc(\Mc) = \set{B_1, \dots, B_n}$ and $\scc(\Mc) \setminus \bscc(\Mc) = \set{S_1, \dots, S_m}$ with $S_i$ in reverse topological order.
		%\State Set $\gain \gets \emptyset$, $\bias \gets \emptyset$
		\For {$B_i \in \bscc(\Mc)$} \Comment{Obtain gain and bias of BSCCs}
			\State Obtain $\gain_i$ and $\bias(s)$ for all $s \in B_i$ by solving the equations \label{alg:scc_eval:bscc_eval}
			\begin{equation*}
				\begin{aligned}
					\bias(s) & = \expsucc^{B_i}(\bias, s) + \rew(s) - \gain_i \quad \forall s \in B_i, \\
					\bias(s_i) & = 0 \quad \text{for one arbitrary but fixed $s_i \in B_i$}.
				\end{aligned}
			\end{equation*}
			\State Set $\gain(s) \gets g_i$ for all $s \in B_i$.
		\EndFor
		\For{$i$ from $1$ to $m$} \Comment{Obtain gain and bias of non-BSCC states}
			\State Let $\states^< := \Union_{j=1}^{i-1} S_j \union \Union_{j=1}^n B_j$
			\State Compute $\mathsf{succ}(\mathsf{\gain}) \gets \set{s' \in \states^< \mid \exists s \in S_i.~\trans(s, s') > 0 \land \gain(s') = \mathsf{\gain}}$.
			\State Set $\mathsf{succ\gain} = \set{\mathsf{\gain} \mid \mathsf{succ}(\mathsf{\gain}) \neq \emptyset}$. \label{alg:scc_eval:gain_values}
			\State For each $\mathsf{\gain} \in \mathsf{succ\gain}$, obtain $p_{\mathsf{\gain}}$ by solving the equations \label{alg:scc_eval:gain_reach}
			\begin{equation*}
				p_{\mathsf{\gain}}(s) = \expsucc^{S_i}(p_{\mathsf{\gain}}, s) + \sum_{s' \in \mathsf{succ}{\mathsf{\gain}}} \trans(s, s') \quad \forall s \in S_i.
			\end{equation*}
			\State Set $\gain(s) \gets \sum_{\mathsf{\gain} \in\mathsf{succ\gain}} p_{\mathsf{\gain}}(s) \cdot \mathsf{\gain}$ for all $s \in S_i$. \label{alg:scc_eval:gain_value}
			\State Obtain $\bias(s)$ for all $s \in S_i$ by solving the equations \label{alg:scc_eval:bias_values}
			\begin{equation*}
				\bias(s) = \expsucc^{S_i}(\bias, s) + \expsucc^{\states^<}(\bias, s) + \rew(s) - \gain(s) \quad \forall s \in S_i.
			\end{equation*}
		\EndFor
		\State \Return $(\gain, \bias)$.
	\end{algorithmic}
\end{algorithm}

As with the compression approach, we exploit the fact that all states in some BSCC have the same gain.
But instead of encoding this information into one big linear equation system, we use it to obtain multiple sub-problems.

First, we obtain gain and bias for each BSCC separately in Line~\ref{alg:scc_eval:bscc_eval}.
Note that there are only $\cardinality{B_i} + 1$ variables and equations, since there only is a single gain variable.
The last equation, setting bias to zero for some state of the BSCC, again induces a unique solution.

Now, these values are back-propagated through the MC.
As mentioned, we can obtain a topological ordering of the SCCs, where a state $s$ in a \enquote{later} SCC cannot reach any state $s'$ in some earlier SCC.
By processing the SCCs in reverse topological order, we can successively compute values of all states as follows.

Since the gain actually is only earned in BSCCs, the gain of some non-BSCC state naturally only depends on the probability of ending up in some BSCC.
More generally, by a simple inductive argument, the gain of such a non-BSCC state only depends on the gains of the states it ends up in after moving to a later SCC.
In other words, the gain only depends on the reachability of the successor gains.
So, instead of constructing a linear equation system involving both gain and bias for each SCC, we determine the different \enquote{gain outcomes} in Line~\ref{alg:scc_eval:gain_values} and then compute the probability of these outcomes in Line~\ref{alg:scc_eval:gain_reach}, i.e.\ the probability of reaching a state obtaining some particular successor gain.
Finally, we simply set the gain of some state as the expected outcome in Line~\ref{alg:scc_eval:gain_value}.
Only then the bias is computed in Line~\ref{alg:scc_eval:bias_values} by solving the bias equation with the computed gain values inserted as constants.

At first glance, this might seem rather expensive, as there are $\cardinality{\mathsf{succ\gain}} + 1$ linear equation systems instead of one.
But the corresponding matrices of the systems in Lines~\ref{alg:scc_eval:gain_reach} and \ref{alg:scc_eval:bias_values} actually are (i)~square with a unique solution, allowing the use of LU decomposition; and (ii)~are the same for a particular SCC, enabling reuse of the obtained decomposition.
(For proof, see Appendix~\ref{sec:proof_scc_eval}.)%FINAL
%Now, the equation systems for each SCC are of the size $\cardinality{S_i}^2$ and thus usually smaller than $\cardinality{\states}^2$.

Note how this in fact generalizes the idea of computing attractors in the BSCC-compression approach.
Suppose a non-BSCC state $s \in S_j$ is in the attractor of a particular BSCC $B_i$.
Since moving to $B_i$ is the only possible outcome, $\mathsf{succ\gain}$ as computed in Line~\ref{alg:scc_eval:gain_values} actually is a singleton set containing only the gain $\gain_i$ of the BSCC.
Then $p_{\gain_i}(s) = 1$ for all states in $S_j$ and we can immediately set $\gain(s) = \gain_i$.

\section{Approximation-guided solutions} \label{sec:heuristic_si}
\secspace %

This section introduces another idea to increase efficiency of the strategy iteration.
Section~\ref{sec:mec_heuristics} then combines this method with optimizations of the previous section in a non-trivial way, yielding a super-additive optimization effect.
Our new approach relies on the following observation.
In order to improve a strategy, it is not always necessary to know the exact gain in each state; sufficiently tight bounds are enough to decide that the current action is sub-optimal.
To this end, we assume that we are given an approximative oracle for the gain of any state under some strategy\footnote{We will go into detail why we do not deal with bias later on.}.
Formally, we require a function $\gain^\approx : \straas^{\mathsf{MD}} \times \states \to \Reals^{\geq 0} \times \Reals^{\geq 0}$ and call it \emph{consistent} if for $\gain^\approx(\straa, s) = (\gain_L(\straa, s), \gain_U(\straa, s))$ we have that $\gain^\straa(s) \in [\gain_L(\straa, s), \gain_U(\straa, s)]$.
For readability, we write $\gain_L(\straa)$ and $\gain_U(\straa)$ for the functions $s \mapsto \gain_L(\straa, s)$ and $s \mapsto \gain_U(\straa, s)$, respectively.

\begin{algorithm}[t]
	\caption{\textsc{Approx-SI}}
	\label{alg:heuristic_si}
	\begin{algorithmic}[1]
		\Require MDP $\Mdp = (\states, \initstate, \actions, \av, \trans, \rew)$ and consistent gain approximation $\gain^\approx$.
		\Ensure $(g^*, \pi^*)$, s.t.\ $g^*$ is the optimal gain of the MDP and is obtained by $\pi^*$.
		\State Set $n \gets 0$, and pick an arbitrary strategy $\straa_0 \in \straas^{\mathsf{MD}}$.
		\State Set $\straa_{n+1} = \straa_n$ \label{alg:heuristic_si:line:gain_approx_eval}
		\For {$s \in S$} \label{alg:heuristic_si:line:approx_gain_start}
			\Comment{Approximate gain improvement}
			\If {$\gain_U(\straa_n, s) < \max_{a \in \av(s)} \expsucc(\gain_L(\straa_n), s)$} \label{alg:heuristic_si:line:approx_gain_nonopt}
				\State Pick $\straa_{n+1} \in \argmax_{a \in \av(s)} \expsucc(\gain_L(\straa_n), s, a)$.\label{alg:heuristic_si:line:approx_gain_nonopt_update}
			\EndIf
		\EndFor \label{alg:heuristic_si:line:approx_gain_end}
		\If {$\straa_{n+1} \neq \straa_n$} increment $n$ by 1, go to Line~\ref{alg:heuristic_si:line:gain_approx_eval}.
		\EndIf
		\State Obtain $\gain_{n+2}$ and $\straa_{n+2}$ by one step of precise SI. \label{alg:heuristic_si:line:precise}
		\Comment{Precise improvement}
		\If {$\straa_{n+2} \neq \straa_{n+1}$} increment $n$ by 2, go to Line~\ref{alg:heuristic_si:line:gain_approx_eval}.
		\EndIf
		\State \Return $(\gain_{n+2}, \straa_{n+2})$ \label{alg:heuristic_si:line:return_after_precise}
	\end{algorithmic}
\end{algorithm}

In Algorithm~\ref{alg:heuristic_si}, we define a variant of strategy iteration, which incorporates this approximation for gain improvement.
Let us focus on this improvement in Line~\ref{alg:heuristic_si:line:approx_gain_nonopt_update}.
There are three cases to distinguish.
(1) If the test on Line~\ref{alg:heuristic_si:line:approx_gain_nonopt} holds, i.e.\ the upper bound on the gain in the current state is smaller than the lower bound under some other action $a$, then $a$ definitely gives us a better gain.
Therefore, we switch the strategy to this action.
If the test does not hold, there are two other cases to distinguish:
(2) If in contrast, the lower bound on the gain in the current state is bigger than the upper bound under any other action, the current gain definitely is better than the gain achievable under any other action.
Hence the current action is optimal and the strategy should not be changed.
(3) Otherwise, the approximation does not offer us enough information to conclude anything.
The current action is neither a clear winner nor a clear loser compared to the other actions.
In this case we also refrain from changing the strategy.
Intuitively, if there are any changes to be done in Case (3), we postpone them until no further improvements can be done based solely on the approximations.
They will be dealt with in Line~\ref{alg:heuristic_si:line:precise}, where we determine the gain precisely.
\begin{theorem} \label{stm:heuristic_si_correct}
	Algorithm~\ref{alg:heuristic_si} terminates for any MDP and consistent gain approximation function.
	Furthermore, the gain and corresponding strategy returned by the algorithm is optimal.
\end{theorem}
The proof can be found in Appendix~\ref{sec:proof_heuristic_si_correct}. %FINAL

\subsubsection{Implementing gain approximations}
In order to make Algorithm \ref{alg:heuristic_si} practical, we provide a prototype for such a gain approximation.
To this end, we can again interpret the MC $\Mc$ as a degenerate MDP $\Mdp$ and apply variants of the value iteration methods of \cite[Algorithm 2]{cav}.
We want to emphasize that there are no restrictions on the oracle except consistency, hence there may be other, faster methods applicable here.
This also opens the door for more fine-tuning and optimizations.
For instance, instead of \enquote{giving up} on the estimation and solving the equations precisely, the gain approximation could be asked to refine the estimate for all states where there is uncertainty and Case~(3) occurs.

\subsubsection{Difficulties in using bias estimations}
One may wonder why we did not include a bias estimation function in the previous algorithm.
There are two main reasons for this, namely (i)~by naively using the bias approximation, the algorithm may not converge any more (even with $\varepsilon$-precise approximations) and (ii)~it seems rather difficult to efficiently obtain a reasonable bias estimate.
We provide more detail and intuition in Appendix~\ref{sec:bias_approximation_difficult}. %FINAL

\subsection{Synergy of the approaches} \label{sec:mec_heuristics}

\begin{algorithm}[t]
	\makeatletter \renewcommand{\ALG@name}{Procedure} \makeatother
	\caption{\textsc{MEC-Approx}}
	\label{alg:heuristic_local_si_gain}
	\begin{algorithmic}[1]
		\State Set $\gain^{\max}_L(\straa_n) \gets \max_{s \in M} \gain_L(\straa_n, s)$, $S_- \gets \set{s \mid \gain_U(\straa_n, s) < \gain^{\max}_L(\straa_n)}$, $S_+ = M \setminus S_-$. \label{alg:heuristic_local_si_gain:line:gain_eval}
		\If {$S_- = \emptyset$} Continue with precise improvement.
		\Else
			\While {$S_- \neq \emptyset$}
				\State Obtain $s \in S_-$ and $a \in \av(s)$ such that $\sum_{s' \in S_+} \trans(s, a, s') > 0$.
				\State Set $\straa_{n+1}(s) \gets a$, $S_+ \gets S_+ \union \set{s}$, $S_- \gets S_- \setminus \set{s}$.
			\EndWhile
			\State Increment $n$ by 1, go to Line~\ref{alg:heuristic_local_si_gain:line:gain_eval}.
		\EndIf
	\end{algorithmic}
\end{algorithm}

In order to further improve the approximation-guided approach, we combine it with the idea of MEC decomposition, which in turn allows for even more optimizations.
As already mentioned, each state in a MEC has the same optimal gain.
In combination with the idea of the algorithm in \cite[Sec.~9.5.1]{Puterman}, this allows us to further enhance the gain improvement step as follows.

The gain $\gain^*(M)$ of some MEC $M$ certainly is higher than the lower bound achieved through some strategy in any state of the MEC, which is $\gain^{\max}_L(\straa_n) := \max_{s \in M} \gain_L(\straa_n, s)$.
Hence, any state of the MEC which has an upper bound less than $\gain^{\max}_L(\straa_n)$ is suboptimal, as we can adapt the strategy such that it achieves at least this value in every state of the MEC.
With this, the gain improvement step can be changed to (i)~determine the maximal lower bound $\gain^{\max}_L(\straa_n)$, (ii)~identify all states $S_+$ which have an upper bound greater than this lower bound and (iii)~update the strategy in all other states $S_-$ to move to this \enquote{optimal} region.
Algorithm~\ref{alg:heuristic_local_si_gain} then is obtained by replacing the approximate gain improvement in Lines~\ref{alg:heuristic_si:line:approx_gain_start} to \ref{alg:heuristic_si:line:gain_approx_eval} of Algorithm~\ref{alg:heuristic_si} by Procedure~\ref{alg:heuristic_local_si_gain}.
\begin{theorem} \label{stm:heuristic_local_si_correct}
	Algorithm~\ref{alg:heuristic_local_si_gain} terminates for any MDP and consistent gain approximation function.
	Furthermore, the gain and corresponding strategy returned by the algorithm indeed is optimal.
\end{theorem}
The proof can be found in Appendix~\ref{sec:proof_heuristic_local_si_correct}. %FINAL %
\section{Experimental Evaluation} \label{sec:exper}

In this section, we compare the presented approaches to established tools.

\subsubsection{Implementation details} We implemented our constructions\footnote{Accessible at \url{https://www7.in.tum.de/~meggendo/artifacts/2017/atva_si.txt}} in the PRISM Model Checker \cite{KNP11}.
We also added several general purpose optimizations to PRISM, improving the used data structures.
This may influence the comparability of these results to other works implemented in PRISM.

In order to solve the arising systems of linear equations, we used the \texttt{ojAlgo} Java library\footnote{\url{http://ojalgo.org/}}.
Whenever possible, we employed LU decomposition to solve the equation systems and SVD otherwise.
We use \texttt{double} precision for all computations, which implies that results are only precise modulo numerical issues.
The implementation can easily be extended to arbitrary precision, at the cost of performance.
Further, our implementation only uses the parallelization of \texttt{ojAlgo}.
Since the vast majority of computation time is consumed by solving equation systems, we did not implement further parallelization.

\subsubsection{Experimental setup} All benchmarks have been run on a \texttt{4.4.3-gentoo} x64 virtual machine with 16 cores of 3.0 GHz each, a time limit of 10 minutes and memory limit of 32 GB, using the 64-bit Oracle JDK version \texttt{1.8.0\_102-b14}.
All time measurements are given in seconds and are averaged over 10 executions.
Instead of measuring the time which is spent in a particular algorithm, we decide to measure the overall \emph{user CPU time} of the PRISM process using the UNIX tool \texttt{time}.
This metric has several advantages.
It allows for an easy and fair comparison between, e.g., multithreaded executions, symbolic methods or implementations which do not construct the whole model.
Further, it reduces variance in measurements caused by the operating system through, e.g., the scheduler.
Note that this also allows for measurements larger than the specified timeout, as the process may spend this timeout on each of the 16 cores.
Also, we want to mention that our comparisons would profit from measuring real time, since the majority of SI computations is carried out in parallel, whereas the tools we compare to are hardly parallelized.

\subsection{Models}

We briefly explain the examples used for evaluation.
\textbf{virus} \cite{kwiatkowska2009probabilistic} models a \emph{virus spreading through a network}.
We reward each attack carried out by an infected machine.
\textbf{cs\_nfail} \cite{komuravelli2012assume} models a \emph{client-server mutual exclusion protocol} with probabilistic failures of the clients.
A reward is given for each successfully handled connection.
\textbf{phil\_nofair} \cite{DFP04} represents the (randomised) \emph{dining philosophers} without fairness assumptions.
We use two reward structures, rewarding \enquote{thinking} and \enquote{eating}, respectively.
\textbf{sensor} \cite{komuravelli2012assume} models a \emph{network of sensors} sending values to a central processor over a lossy connection.
Processing received data is rewarded.
\textbf{mer} \cite{Feng2011} captures the behaviour of a \emph{resource arbiter} on a Mars exploration rover.
We reward each time some user is granted access to a resource by the arbiter.

\subsection{Tools}

Since we are unaware of other implementations, we implemented standard SI as in Algorithm~\ref{alg:si} by ourselves.
We compare the following variants of SI.
\begin{itemize}
	\item \texttt{SI}: Standard SI as presented in Algorithm~\ref{alg:si}.
	\item \texttt{BSCC}: SI with BSCC compression gain/bias equations.
	\item \texttt{SCC}: The SCC decomposition approach of Algorithm~\ref{alg:scc_eval}.
	\item $\texttt{SCC}_A$: The SCC decomposition approach combined with the approximation methods from Section~\ref{sec:heuristic_si}.
\end{itemize}
Further, a \enquote{$M$} superscript denotes use of the MEC decomposition approach as in Algorithm~\ref{alg:mec_dec_si}.
In the case of $\texttt{SCC}_A^M$, we use the improved method of Section~\ref{sec:mec_heuristics}.
More details and evaluation of some further variants can be found in Appendix~\ref{sec:further_experiments}. %FINAL
During our experiments, we observed that the algorithm used to solve the resulting reachability problem did not influence the results significantly, since the weighted quotients are considerably simpler than the original models.

We compare our methods to the value iteration approach we presented in \cite[Algorithm 2]{cav} with a required precision of $10^{-8}$ (\texttt{VI}).
This comparison has to be evaluated with care, since (i)~value iteration inherently is only $\varepsilon$-precise and (ii)~it needs a MEC decomposition for soundness.
Note that topological optimizations for value iteration as suggested by, e.g., \cite{BKLPW17} are partially incorporated by \texttt{VI}, since each MEC is iterated separately.

We also provide a comparison to the LP-based MultiGain~\cite{DBLP:conf/tacas/BrazdilCFK15} in Appendix~\ref{sec:further_experiments}. %FINAL
In summary, the LP approach is soundly beaten by our optimized approaches.
A more detailed comparison can be found in~\cite{cav}.

We are unaware of other implementations capable of solving the mean payoff objective for MDPs.
Neither did we find a mean payoff solver for stochastic games which we could easily set up to process the PRISM models.

\subsection{Results}

\begin{table}[t]
	\centering
	\caption{Comparison of various variants on the presented models.
		Timeouts and memouts are denoted by a hyphen.
		The best results in each row are marked in bold, excluding \texttt{VI}.
		The number of states and MECs are written next to the model.}
	\setlength\tabcolsep{0pt}
	\begin{tabu} to 0.99\linewidth {l@{\hspace{2pt}}|CCCCCCC|C}
		\textbf{Model} & \texttt{SI} & $\texttt{SI}^M$ & \texttt{BSCC} & $\texttt{BSCC}^M$ & \texttt{SCC} & $\texttt{SCC}_A$ & $\texttt{SCC}^M_A$ & \texttt{VI} \\
		\toprule
		% MODEL                            SI         SI^M    BSCC       BSCC^M            SCC         SCC_A       SCC^M_A    VI
		cs\_nfail3 (184, 38)    &          17 & \textbf{4} & \textbf{4} & \textbf{4} &   \textbf{4} &  \textbf{4} &  \textbf{4} &   4 \\
		cs\_nfail4 (960, 176)   &        1129 &          6 &         16 & \textbf{5} &   \textbf{5} &  \textbf{5} &           6 &   5 \\
		\midrule
		virus (809, 1)          &         $-$ & \textbf{4} &         10 & \textbf{4} &            5 &           5 &  \textbf{4} &   4 \\
		\midrule
		phil\_nofair3 (856, 1)  &         $-$ &        $-$ &        112 &        112 &   \textbf{6} &          10 &           7 &   5 \\
		phil\_nofair4 (9440, 1) &         $-$ &        $-$ &        $-$ &        $-$ &  \textbf{15} &         310 &         107 &  18 \\
		\midrule
		sensors1 (462, 132)     &         $-$ &         13 &         23 & \textbf{4} &   \textbf{4} &           6 &  \textbf{4} &   5 \\
		sensors2 (7860, 4001)   &         $-$ &         89 &        $-$ &         14 &           13 &         168 & \textbf{11} &  15 \\
		sensors3 (77766, 46621) &         $-$ &        $-$ &        $-$ &         78 &  \textbf{40} &         $-$ &          46 &  72 \\
		\midrule
		mer3 (15622, 9451)      &         $-$ &         21 &        $-$ &         26 &  \textbf{16} &         244 &          22 &  15 \\
		mer4 (119305, 71952)    &         $-$ &         58 &        $-$ &        163 &  \textbf{42} &         $-$ &          84 &  64 \\
		mer5 (841300, 498175)   &         $-$ &        $-$ &        $-$ &        $-$ & \textbf{474} &         $-$ &         $-$ & $-$ \\
		\bottomrule
	\end{tabu}
	\label{tbl:experiments}
\end{table}

We will highlight various conclusions to be drawn from Table~\ref{tbl:experiments}.
Comparing the naive SI with our enhanced versions \texttt{BSCC} and \texttt{SCC}, the number of strategy improvements does not differ, but the evaluation of each strategy is significantly faster, yielding the differences displayed in the table.

On the smaller models (\textbf{cs\_nfail} and \textbf{virus}) nearly all of the optimized methods perform comparable, a majority of the runtime actually is consumed by the start-up of PRISM.
Especially on \textbf{virus}, all the MEC-decomposition approaches have practically the same execution time due to the model only having a single MEC with a single state, which makes solving the model trivial for these approaches.

The results immediately show how intractable naive strategy iteration is.
On models with only a few hundred states, the computation already times out.
The BSCC compression approach \texttt{BSCC} suffers from the same issues, but already performs significantly better than \texttt{SI}.
In particular, when combined with MEC decomposition, it is able to solve more models within the given time.

Further, we see immense benefits of using the SCC approach, regularly beating even the quite performant (and imprecise) value iteration approach.
Interestingly, the variants using approximation often perform worse than the \enquote{pure} SCC method.
We conjecture that this is due the gain approximation function we used.
It computes the gain up to some adaptively chosen precision instead of computing up to a certain number of iterations.
Changing this precision bound gave mixed results, on some models performance increased, on some it decreased.
Comparing the two approximation-based approaches $\texttt{SCC}_A$ and $\texttt{SCC}_A^M$, we highlight the improvements of Algorithm~\ref{alg:heuristic_local_si_gain}, speeding up convergence even though a MEC decomposition is computed.

Finally, we want to emphasize the \texttt{mer} results.
Here, our \texttt{SCC} approach manages to obtain a solution within the time- and memory-bound, while all MEC decomposition approaches fail due to a time-out. %

\section{Conclusion}

We have proposed and evaluated several techniques to speed up strategy iteration.
The combined speed ups are in orders of magnitude.
This makes strategy iteration competitive even with the most used and generally imprecise value iteration and shows the potential of strategy iteration in the context of MDPs.

In future work, we will further develop this potential.
Firstly, building upon the \emph{SCC decomposition}, we can see opportunities to interleave the SCC computation and the improvements of the current strategy.
Secondly, the \emph{gain approximation} technique used is quite naive.
Here we could further adapt our recent results on VI \cite{cav}, in order to improve the performance of the approximation.
Besides, we suggest to use simulations to evaluate the strategies.
Nevertheless, the incomplete confidence arising form stochastic simulation has to be taken into account here.
Thirdly, techniques for efficient \emph{bias} approximation and algorithms to utilize it would be desirable.
Finally, a fully configurable tool would be helpful to find the sweet-spot combinations of these techniques and useful as the first scalable tool for mean payoff optimization in MDPs.

\paragraph{Acknowledgments} We thank the anonymous reviewers for their insightful comments and valuable suggestions.
In particular, a considerable improvement to the BSCC compression approach of Section~\ref{sec:sccs} has been proposed.

\bibliographystyle{alpha}%FINAL\bibliographystyle{abbrv}
\bibliography{ref}

\vfill
\pagebreak
\appendix
\section*{Appendix}
\section{Linear Algebra} \label{sec:linear_algebra}

We consider quite a few linear equation systems, i.e.\ equations of the form $A \cdot x = b$, where $A \in \Reals^{n \times m}$ is some matrix, $x \in \Reals^m$ is a solution vector to be determined and $b \in \Reals^n$ are the constant terms of the system.
To this end, we quickly recall some basic terminology related to this problem.
For an in depth discussion of these topics, we refer to the numerous existing books, e.g.\  \cite{cheney2012numerical}.

\paragraph{Condition number}
The condition number of a matrix $A$ intuitively describes how much the norm of $A \cdot x$ changes depending on $x$.
It is directly related to the rate of convergence and numerical stability of many solution methods for linear equation systems.
As a rule of thumb, a condition number of $\kappa = 10^k$ roughly translates to losing up to $k$ digits of accuracy \cite[p. 321]{cheney2012numerical}.

\paragraph{Solving linear equations}
There are many different methods to solve linear equation systems precisely.
Most of the precise method are so called decomposition approaches, where the majority of computational effort goes into decomposing $A$ into some other matrices.
Once such a decomposition is obtained, solving the equation system for multiple $b$ is very quick.
Our implementation uses \emph{singular value decomposition (SVD)}, which exists for any matrix, and \emph{LU decomposition}, which only exists for full-rank square matrices, but is considerably faster.

\section{Advantage of strategy iteration over value iteration} \label{sec:strategy_faster_than_value_iteration}

\begin{figure}[h]
	\centering
	\begin{tikzpicture}[auto,initial text={},node distance=1cm]
		\node[state,initial] (state_1) {$s_1$};
		\node[right=of state_1, draw,shape=circle,fill=black,scale=0.2] (state_1_to_2) {};
		\node[state,right=of state_1_to_2] (state_2) {$s_2$};
		\node[right=of state_2, draw,shape=circle,fill=black,scale=0.2] (state_2_to_dots) {};
		\node[right=of state_2_to_dots] (dots) {$\cdots$};
		\node[state,right=of dots] (state_n) {$s_n$};

		\path[->]
		(state_1) edge[-]                        node {$a, 0$} (state_1_to_2)
		(state_1_to_2) edge[out=-90,in=-30,looseness=0.95] node {} (state_1)
		(state_1_to_2) edge                      node {$0.01$} (state_2)
		(state_2) edge[-]                        node {$a, 0$} (state_2_to_dots)
		(state_2_to_dots) edge[out=-120,in=-30,looseness=0.5,pos=0.75] node {$0.99$} (state_1)
		(state_2_to_dots) edge                   node {$0.01$} (dots)
		(dots) edge                              node {$0.01$} (state_n)
		(state_n) edge[loop right]               node {$b, 1$} (state_n)
		(state_n) edge[in=70,out=110,looseness=0.3,swap] node {$a, 0$} (state_1);

	\end{tikzpicture}
	\caption{A small example highlighting why strategy iteration performs better for some models.
		On each edge we write the action corresponding to this transition and the reward for taking this action, followed by the probabilistic branching, if any.}
	\label{fig:strategy_faster_than_value_iteration}
\end{figure}
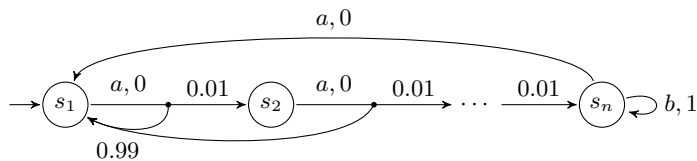
\begin{example} \label{example:strategy_faster_than_value_iteration}
	Consider the MDP given in Figure~\ref{fig:strategy_faster_than_value_iteration}.
	The optimal gain is $1$ in all states.
	When solving this model using value iteration, the algorithm takes exponentially many steps to yield this result, since the back-propagation of the value $1$ from $s_n$ is slowed down by a factor of $0.01$ in each state.
	Especially, after $n$ iterations, the value of $s_1$ would only be $0.01^n$.
	
	In contrast, strategy iteration only needs at most one improvement steps to achieve this result.
	Starting with the strategy which chooses action $a$ in each state, strategy iteration identifies a suboptimal choice in state $s_n$ by the evaluation of the bias.
	Hence, the improvement step will switch to action $a$, yielding the optimal strategy.
	By again determining gain and bias of the second strategy, the algorithm verifies this optimality and terminates.

	We confirmed this intuitional reasoning by experimental evaluation.
	For $n = 500$, the SCC decomposition approach of Algorithm~\ref{alg:scc_eval} terminates within a few seconds and even for $n = 5000$ the computation completes within two minutes.
	In comparison, value iteration already takes a noticeable amount of time for $n = 5$ and even fails to yield a result for $n = 10$ after an hour.
\end{example}

\section{Necessity of bias improvement} \label{sec:bias_improvement_necessary}

\begin{figure}[h]
	\centering
	\begin{tikzpicture}[auto,initial text={}]
		\node[state,initial] (state_0) {$s_1$};
		\node[state,right=2cm of state_0] (state_1) {$s_2$};
		
		\path[->]
		(state_0) edge[loop above]       node {$a, 1$} (state_0)
		(state_0) edge[out=25, in=155]   node {$b, 3$} (state_1)
		(state_1) edge[loop above]       node {$a, 1$} (state_1)
		(state_1) edge[out=205, in=-25]  node {$b, 3$} (state_0);
	\end{tikzpicture}
	\caption{A small example highlighting why bias improvement is necessary for strategy iteration.
		On each edge we write the action corresponding to this transition and the reward for taking this action.}
	\label{fig:bias_improvement_necessary}
\end{figure}
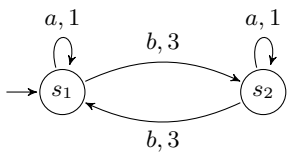
\begin{example} \label{example:bias_improvement_necessary}
	Consider the MDP given in Figure~\ref{fig:bias_improvement_necessary}
	Clearly, the optimal gain is $3$ in both states, obtained by the strategy $\straa^* = (b, b)$, i.e.\ taking action $b$ in both states.
	Suppose the algorithm instead starts with the strategy $\straa = (a, a)$.
	Then, $\gain(s_1) = \gain(s_2) = 1$ and the gain improvement step in Algorithm~\ref{alg:si}, Line~\ref{alg:si:line:gain_improve} does not modify the strategy.
	For the bias we have that $\bias(s_1) = \bias(s_2) = 0$ and thus for both states $\rew(s, a) + \expsucc(\bias, s, a) = 1 < 3 = \rew(s, b) + \expsucc(\bias, s, b)$.
	Hence, the bias improvement step changes the strategy to the optimal $\straa^*$.
\end{example}

\section{Weighted MEC quotient} \label{sec:weighted_quotient}

In this section, we recall the \emph{weighted MEC quotient} from~\cite{cav}.
This construction is a variant of the MEC quotient defined in~\cite{DeAlfaro1997} and is used to reduce mean payoff to a reachability query, given that we obtained the gain of each MEC $\gain^*(M)$.
We provide a formal definition and an intuition for the reduction.
\begin{definition}
	Let $\Mdp = (\states,\initstate,\actions,\av,\trans,\rew)$ be an MDP with $\mec(\Mdp) = \set{M_1, \dots, M_n}$ where $M_i = (T_i, A_i)$, and $\mec_\states = \Union_{i=1}^n T_i$ the set of all states in some MEC.
	Further, let $f: \mec(\Mdp) \mapsto [0,1]$ be a function which assigns a value to each MEC $M_i$.
	The \emph{weighted MEC quotient of $\Mdp$ and $f$} is the MDP $\Mdp^f = (\states^f, \initstate^f, \actions^f, \av^f, \trans^f, \rew^f)$ defined as follows.
	\begin{itemize}
		\item $\states^f = \states \setminus \mec_\states \union \set{\widehat{s}_1, \dots, \widehat{s}_n} \union \set{s_+, s_-}$.
		\item If for some $1 \leq i \leq n$ we have $\initstate \in T_i$, then $\initstate^f = \widehat{s}_i$; otherwise $\initstate^f = \initstate$.
		\item $\actions^f =  \set{(s, a) \mid s \in \states, a \in \av(s)} \union \set{\mathrm{stay}}$.
		\item $\av^f$ is defined as:
		\begin{align*}
			\forall s \in \states \setminus \mec_\states.~& \av^f(s) = \set{ (s, a) \mid a \in \av(s) }, \\
			\forall 1 \leq i \leq n.~& \av^f(\widehat{s}_i) = \set{ (s, a) \mid s \in T_i \land a \in \av(s) \setminus A_i } \union \set{\mathrm{stay}}, \\
			& \av^f(s_+) =\av^f(s_-) = \emptyset,
		\end{align*}
		\item $\trans^f$ is defined as follows.
		For $s \in \states \setminus \mec_\states$ we define the transition function just as in the original MDP
		\begin{equation*}
			\forall s \in \states \setminus \mec_\states, a \in \av(s), t \in \states.~\trans^f(s, (s, a), t) = \trans(s, a, t).
		\end{equation*}
		For the MEC representative states $\widehat{s}_i$, we instead distinguish multiple cases for the target state.
		To this end, let $(s, a) \in \av(\widehat{s}_i)$. For $t \in \states^f \setminus \mec_\states$, we define
		\begin{equation*}
			\trans^f(\widehat{s}_i, (s, a), t) = \trans(s, a, t).
		\end{equation*}
		If instead $t = \widehat{s}_j$, i.e.\ the representative of another MEC, we define
		\begin{equation*}
			\trans^f(\widehat{s_i}, (s, a), \widehat{s}_j) = \sum_{s' \in T_j} \trans(s, a, s'). \vspace{-0.5em}
		\end{equation*}
		Now, only the special case of the $\mathrm{stay}$ action remains, which we define as $\trans(\widehat{s_i}, \mathrm{stay}) = \set{s_+ \mapsto f(M_i), s_- \mapsto 1 - f(M_i)}$.
		\item $\forall s \in \states^f, a \in \av^f(s).~\rew^f(s,a) = 0$.
	\end{itemize}
\end{definition}
Intuitively, the weighted quotient is obtained by the following steps.
First, all states in the same MEC are merged into one representative state while preserving all transitions between different MECs.
Furthermore, special states $s_+$ and $s_-$ are added to the quotient.
On each state corresponding to some MEC $M$, a distinct \emph{stay} action with transitions to $s_+$ and $s_-$ is available.
This action corresponds to \enquote{committing} to this MEC and acquiring the value $f(M)$.
To achieve that, the transition probabilities to these two special states are chosen proportional to $f(M)$.
Intuitively, reaching $s_+$ corresponds to obtaining a value of 1 and dually $s_-$ corresponds to obtaining 0.
With this construction, we are able to express the maximal obtainable $f(M)$
\begin{equation*}
	\max_{\straa \in \straas^{\mathsf{MD}}} \sum_{M \in \mec(\Mdp)} \pr^\straa_s[\Diamond \Box M] \cdot f(M) \vspace{-0.5em}
\end{equation*}
as the reachability of $s_+$ in $\Mdp^f$, i.e.\ $\max_{\pi \in \straas^{\mathsf{MD}}} \pr_s^\straa[\Diamond s_+]$ \cite{cav}.
In our specific case, we define $f(M) = \gain^*(M) / \rmax$ with $\rmax = \max_{s \in \states, a \in \av(s)} \rew(s, a)$ and rescale the resulting reachability by $\rmax$, giving us the desired result.

\section{Proofs} \label{sec:proofs}

\subsection{Proofs for Algorithm~\ref{alg:scc_eval}} \label{sec:proof_scc_eval}

We first show that all the linear equation systems occurring in the algorithm are square and have a unique solution.
\begin{proof}
	There are three types of equation systems in the algorithm, (i)~the BSCC evaluation in Line~\ref{alg:scc_eval:bscc_eval}, (ii)~the gain-reachability in Line~\ref{alg:scc_eval:gain_reach} and finally (iii)~the bias determination in Line~\ref{alg:scc_eval:bias_values}.
	The equation systems of type (i) involve $1 + \cardinality{B_i}$ variables, namely one gain variable for the whole BSCC and one bias variable per state and also $1 + \cardinality{B_i}$ equations, giving a square matrix.
	Uniqueness follows from Condition~9.2.3 in \cite[Sec.~9.2.1]{Puterman}.

	Similarly, the systems of type (ii) and (iii) involve $\cardinality{S_i}$ variables and $\cardinality{S_i}$ equations.
	Both types are of the form $x(s) = \expsucc^{S_i}(x, s) + \mathsf{rhs}(s)$.
	A solution of (ii) corresponds to the constrained reachability $S_i~\mathsf{U}~\mathsf{succ}(\gain)$ as defined in~\cite[Sec.~10.1.1]{BaierBook}.
	By simple modifications, we conclude uniqueness of the solution by~\cite[Rem.~10.18]{BaierBook}.
	\qed
\end{proof}
Now, we prove correctness of the results returned by the algorithm.
\begin{proof}
	Firstly, we show that all the equations are well defined, i.e.\ any value which is used as a constant in the equation systems was already determined in a previous step.
	Clearly, (i) satisfies this, as $\rew(s)$ is given for all states.
	For equation systems (ii) and (iii) note that as we are only considering states in $\states^<$.
	These have both $\gain$ and $\bias$ set, as we required reverse topological order of the SCCs.

	To show correctness of the results, assume for contradiction that the returned values $(\gain, \bias)$ are erroneous.
	Furthermore, let $(\gain', \bias')$ be the \emph{unique} solutions of the gain/bias equations together with Condition~9.2.3 from \cite[Sec.~9.2.1]{Puterman} (numbering the states appropriately).
	By Corollary~\ref{stm:compression_correct}, $(\gain, \bias)$ and $(\gain', \bias')$ agree on the BSCCs.
	This means that the error occurred while processing some non-BSSC state.
	Let $S_i$ be the first SCC in which such an error occurs, i.e.\ there is some $s \in S_i$ with $\gain(s) \neq \gain'(s)$ or $\bias(s) \neq \bias'(s)$.

	We now prove that $\gain(s)$ and $\bias(s)$ satisfy the gain/bias equations which yields the contradiction.
	First, note that by the reverse topological order of the SCCs we have $\trans(s, s') = 0$ for each $s' \in \states \setminus (S_i \union \states^<)$ and hence
	\begin{equation*}
		\expsucc(\gain, s) = \expsucc^{S_i}(\gain, s) + \expsucc^{\states^<}(\gain, s).
	\end{equation*}
	Furthermore, with $\mathsf{succs} := \set{s' \in \states^< \mid \exists s \in S_i.~\trans(s, s') > 0} = \Union_{\mathsf{\gain} \in \mathsf{succ\gain}} \mathsf{succ}(\mathsf{\gain})$, we have that
	\begin{align*}
		\expsucc^{\states^<}(\gain, s) & = \sum_{s' \in \states^<} \trans(s, s') \cdot \gain(s') = \sum_{s' \in \mathsf{succs}} \trans(s, s') \cdot \gain(s') \\
				& = \sum_{\mathsf{\gain} \in \mathsf{succ\gain}, s' \in \mathsf{succ}(\mathsf{\gain})} \trans(s, s') \cdot \gain(s') = \sum_{\mathsf{\gain} \in \mathsf{succ\gain}, s' \in \mathsf{succ}(\mathsf{\gain})} \trans(s, s') \cdot \mathsf{\gain} \\
				& = \sum_{\mathsf{\gain} \in \mathsf{succ\gain}} \left( \mathsf{\gain} \cdot \sum_{s' \in \mathsf{succ}(\mathsf{\gain})} \trans(s, s') \right) = \sum_{\mathsf{\gain} \in \mathsf{succ\gain}} \mathsf{\gain} \cdot (p_{\mathsf{\gain}}(s) - \expsucc^{S_i}(p_{\mathsf{\gain}}, s)).
	\end{align*}
	For the last equality, we used the characterization of $p_{\mathsf{\gain}}$ in Line~\ref{alg:scc_eval:gain_reach}.
	Using this, we arrive at
	\begin{align*}
		\expsucc(\gain, s) & = \expsucc^{S_i}(\gain, s) + \sum_{\mathsf{\gain} \in \mathsf{succ\gain}} \mathsf{\gain} \cdot (p_{\mathsf{\gain}}(s) - \expsucc^{S_i}(p_{\mathsf{\gain}}, s)) \\
				& = \sum_{\mathsf{\gain} \in \mathsf{succ\gain}} \mathsf{\gain} \cdot p_{\mathsf{\gain}}(s) + \expsucc^{S_i}(\gain, s) - \sum_{\mathsf{\gain} \in \mathsf{succ\gain}} \mathsf{\gain} \cdot \expsucc^{S_i}(p_{\mathsf{\gain}}, s) \\
				& = \gain(s) + \expsucc^{S_i}(\gain, s) - \sum_{\mathsf{\gain} \in \mathsf{succ\gain}} \mathsf{\gain} \cdot \expsucc^{S_i}(p_{\mathsf{\gain}}, s).
	\end{align*}
	Furthermore, by definition of $\gain(s)$ in Line~\ref{alg:scc_eval:gain_value}
	\begin{align*}
		\expsucc^{S_i}(\gain, s) & = \sum_{s' \in S_i} \gain(s') \cdot \trans(s, s') = \sum_{s' \in S_i} \left( \sum_{\mathsf{\gain} \in \mathsf{succ\gain}} p_{\mathsf{\gain}}(s') \cdot \mathsf{\gain} \right) \cdot \trans(s, s') \\
				& = \sum_{\mathsf{\gain} \in \mathsf{succ\gain}} \mathsf{\gain} \cdot \left( \sum_{s' \in S_i} p_{\mathsf{\gain}}(s') \cdot \trans(s, s') \right) = \sum_{\mathsf{\gain} \in \mathsf{succ\gain}} \mathsf{\gain} \cdot \expsucc(p_{\mathsf{\gain}}, s).
	\end{align*}
	Together, we obtain that $\gain(s) = \expsucc(\gain, s)$.
	By again employing the specific order of the SCCs, we also have that
	\begin{equation*}
		\expsucc(\bias, s) = \expsucc^{S_i}(\bias, s) + \expsucc^{\states^<}(\bias, s).
	\end{equation*}
	Hence, by definition of $\bias(s)$ in Line~\ref{alg:scc_eval:bias_values},
	\begin{equation*}
		\bias(s) = \expsucc(\bias, s) + \rew(s) - \gain(s).
	\end{equation*}
	Together with the uniqueness of the solution, we arrive at $\gain(s) = \gain'(s)$ and $\bias(s) = \bias'(s)$, contradicting the assumption.
	\qed
\end{proof}

\subsection{Proof of Theorem~\ref{stm:heuristic_si_correct} (correctness of Algorithm~\ref{alg:heuristic_si})} \label{sec:proof_heuristic_si_correct}

\begin{proof}
	\emph{Correctness}: Follows trivially from Theorem~\ref{stm:si_correct}, since any returned value is deemed optimal by the precise method used in Line~\ref{alg:heuristic_si:line:precise}.
	
	\emph{Termination}: We apply \cite[Theorem 9.2.6]{Puterman} to show termination of the algorithm.
	This theorem intuitively states that if there are strategies $\straa$ and $\straa'$ where either (i) in some state $\straa'$ improved the gain or (ii) for all states the gain is unchanged and in some state $\straa'$ improves the bias, then the gain never is decreased and, in case (i), the gain is strictly increased in some states or, in case (ii), the gain is not modified but the bias is strictly increased.
	Since there are only finitely many strategies these improvements can only occur finitely often.
	Now note that the modification of $\straa_n$ based on the gain approximation in Line~\ref{alg:heuristic_si:line:approx_gain_nonopt_update} is conservative.
	It only modifies the strategy if the improvement definitely increases the gain.
	Hence the pair $\straa_n, \straa_{n+1}$ satisfies the conditions of this theorem.
	The improvements based on the precise evaluation similarly satisfy these conditions (actually, it is used to prove termination of Algorithm~\ref{alg:si}).
	\qed
\end{proof}

\subsection{Proof of Theorem~\ref{stm:heuristic_local_si_correct} (correctness of Algorithm~\ref{alg:heuristic_local_si_gain})} \label{sec:proof_heuristic_local_si_correct}

\begin{proof}
	Let $\gain_n$ and $\gain_{n+1}$ be the gain of $\straa_n$ and $\straa_{n+1}$, respectively.
	We show that each state $s$ for which we change the strategy has $\gain_n(s) < \gain^{\max}_L(\straa_n)$ and $\gain_{n+1}(s) \geq \gain^{\max}_L(\straa_n)$.
	The claim then follows by the same reasoning used to prove Theorem~\ref{stm:heuristic_si_correct}.
	
	Clearly, for states $s \in S_-$ we have that $\gain_n(s) < \gain^{\max}_L(\straa_n)$ by construction of $S_-$.
	Let now $\Mc$ be the Markov chain induced by $\Mdp$ under $\straa_n$ and $\Mc'$ the one induced by $\straa_{n+1}$.
	By Lemma~\ref{stm:gain_attractor_equal} all states in the attractor $\attractor(B)$ of some BSCC $B$ have the same gain.
	Therefore, for any state $s \in B$ we have that $\max_{s \in \attractor(B)} \gain^{\straa_n}_L(s) \leq \gain^{\straa_n}(s) \leq \min_{s \in \attractor(B)} \gain^{\straa_n}_U(s)$, since $\gain^\approx$ is consistent.
	Hence, $\attractor(B) \intersection S_+ \neq \emptyset$ implies that $\attractor(B) \subseteq S_+$.

	Therefore, we define $\mathcal{B} = \set{B \in \bscc(\Mc) \mid B \subseteq S_+}$ the set of all BSCCs in $S_+$.
	Since no state can achieve strictly more than the upper bound among all BSCCs and the gain approximation is consistent, we have that $\gain^{\max}_L(\straa_n)$ is bounded by this upper bound and thus $\mathcal{B}$ is non-empty.
	As we don't change the strategy for any state in $S_+$ and point all states in $S_-$ towards it, $\mathcal{B}$ will be the set of BSCCs in $\Mc'$.
	Therefore, all states will have a gain of at least $\gain^{\max}_L(\straa_n)$ under strategy $\straa_{n+1}$.
	\qed
\end{proof}

\section{Bias approximations} \label{sec:bias_approximation_difficult}

In order to further improve the performance of the presented approximation method, one might apply the idea of approximate gain improvement to bias improvement, too.
Naively, this would mean changing the strategy based on a bias approximation over all actions which the gain approximation deemed roughly equal.
Unfortunately, this approach has two major problems.

As we cannot determine the set of gain optimal choices precisely, the bias improvement may actually change the strategy to an action with a lower gain.
This may result in switching to a strategy which already occurred, introducing cycling and non-termination of the algorithm.
During our investigation we indeed found a simple example where this happens, even with precise bias values, which we show below in Example~\ref{example:approximate_gain_cycling}.
A simple way to fix this is to only allow finitely many approximation-based bias improvements and eventually switching to the precise method.

Additionally, obtaining reasonably precise estimations for the bias seems tricky.
We give an intuition for this issue to motivate more research in this direction.

It is known that the bias corresponds to the total expected reward under the modified reward function $\rew'(s) = \rew(s) - \gain(s)$.
Moreover, one can pick a \enquote{reference state} for each BSCC and set $\bias(s) = 0$ for this particular state.
Then, the bias can also be computed as total expected reward under $\rew'(s)$ until reaching any of the reference states.
Assuming that we obtained a precise gain value, estimating the total reward would be tractable by a value iteration variant.

But obtaining a precise gain value requires us to solve linear equations involving the bias, too, so this approach is ruled out.
Instead, we have to deal with some $\varepsilon$-precise gain value.
But then, the approximation of $\rew'(s)$ potentially has an $\varepsilon$-error in each state.
This means that an $\varepsilon$-precise gain does not allow us to determine an $\varepsilon$-precise bias value by estimating the aforementioned total expected reward.

A possible idea to remedy this problem would be to estimate the average number of steps $n$ until reaching some of the reference states, which allows us to bound the error.
By then computing a $\frac{\varepsilon}{2 n}$-precise gain value, one could then obtain an $\varepsilon$-precise bias estimate.
Apart from the obvious drawback that $n$ is potentially very huge in some models, we furthermore lose advantages compared to value iteration.
In models where $n$ is small, value iteration converges very fast, since intuitively $n$ corresponds to the \enquote{propagation speed} of values through the model.
One of the main reasons why strategy iteration is considerably faster than value iteration on some models is that solving the equation systems is independent of this propagation speed.

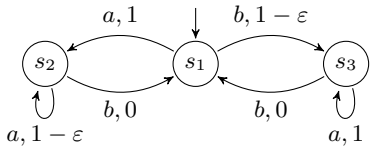
\begin{figure}
	\centering
	\begin{tikzpicture}[auto,initial text={}]
		\node[state,initial above] (state_0) at (0, 0) {$s_1$};
		\node[state] (state_1) at (-2, 0) {$s_2$};
		\node[state] (state_2) at (2, 0) {$s_3$};
		
		\path[->]
		(state_0) edge[bend right,swap]  node {$a, 1$} (state_1)
		(state_0) edge[bend left]        node {$b, 1-\varepsilon$} (state_2)
		(state_1) edge[loop below]       node {$a, 1-\varepsilon$} (state_1)
		(state_1) edge[bend right,swap]  node {$b, 0$} (state_0)
		(state_2) edge[loop below]       node {$a, 1$} (state_2)
		(state_2) edge[bend left]        node {$b, 0$} (state_0);
	\end{tikzpicture}
	\caption{An example MDP used to demonstrate how only using $\varepsilon$-precise gain values for improvement may lead to cycling.
		The notation is the same as in Figure~\ref{fig:bias_improvement_necessary}.}
	\label{fig:approximate_gain_cycling}
\end{figure}
\begin{example}[Potential cycling when using approximate gain values] \label{example:approximate_gain_cycling}
	By this example we demonstrate how using only approximate gain values may lead to cycling of the strategy iteration algorithm, even when the gain values are of some arbitrary precision $\varepsilon$ and precise bias values are available.
	
	To this end, consider the MDP specified in Figure~\ref{fig:approximate_gain_cycling}.
	Note that it consists of a single MEC.
	Consider starting with the strategy $\straa = (a, a, a)$.
	Then the gain approximation might return the values $\gain^\approx(\straa, s_1) = \gain^\approx(\straa, s_2) = (1-2\varepsilon, 1-\varepsilon)$ and $\gain^\approx(\straa, s_3) = (1, 1+\varepsilon)$.
	Thus, the gain improvement step switches to $\straa' = (b, a, a)$.
	
	Now, the gain approximation may return $\gain^\approx(\straa', s_i) = (1-\varepsilon, 1)$ for $i \in \set{1, 2, 3}$.
	As these values are inconclusive, no actions are changed due to gain improvement and the algorithm resorts to bias improvement.
	The precise bias values are given by $\bias_{\straa'}(s_1) = -\varepsilon$ and $\bias_{\straa'}(s_2) = \bias_{\straa'}(s_3) = 0$.
	Since the gain approximation for all states is exactly the same, the algorithm cannot deduce any information about the gain-optimal actions and setting $\av_{\gain}(s) = \set{a, b}$ for all states seems to be the only sensible choice.
	Consequently, the values of the bias improvement condition $r(s, a) + \expsucc(\bias, s, a)$ in state $s_1$ are $1 + 0$ for action $a$ and $(1 - \varepsilon) + 0$ for action $b$.
	Together, the favoured action is $a$.
	In the other two states, the bias improvement dictates to not change the action, since $0 + (-\varepsilon) < (1 - \varepsilon) + 0$ and $0 + (-\varepsilon) < 1 + 0$, respectively.
	Therefore the algorithm switches back to strategy $\straa$, which leads to cycling.
\end{example}

\section{Further experimental results} \label{sec:further_experiments}

\begin{table}[h]
	\centering
	\caption{Comparison of further variants on the presented models.
		We use the same notation as in Table~\ref{tbl:experiments}.
		The standard \texttt{SI} and our best variant $\texttt{SCC}$ are included again for reference.}
	\begin{tabu} to 0.75\linewidth {l|CC|CC}
		\textbf{Model} & \texttt{LP} & $\texttt{SCC}^M$ & \texttt{SCC} & \texttt{SI} \\
		\toprule
		% MODEL                            LP         SCC^M            SCC            SI
		cs\_nfail3 (184, 38)    &  \textbf{2} &           4 &            4 &          17 \\
		cs\_nfail4 (960, 176)   &  \textbf{5} &  \textbf{5} &   \textbf{5} &        1129 \\
		\midrule
		virus (809, 1)          &  \textbf{4} &  \textbf{4} &            5 &         $-$ \\
		\midrule
		phil\_nofair3 (856, 1)  &  \textbf{3} &           6 &            6 &         $-$ \\
		phil\_nofair4 (9440, 1) &          78 &          17 &  \textbf{15} &         $-$ \\
		\midrule
		sensors1 (462, 132)     &  \textbf{3} &           4 &            4 &         $-$ \\
		sensors2 (7860, 4001)   &         101 & \textbf{11} &           13 &         $-$ \\
		sensors3 (77766, 46621) &         $-$ &          53 &  \textbf{40} &         $-$ \\
		\midrule
		mer3 (15622, 9451)      &         $-$ &          20 &  \textbf{16} &         $-$ \\
		mer4 (119305, 71952)    &         $-$ &          54 &  \textbf{42} &         $-$ \\
		mer5 (841300, 498175)   &         $-$ &         $-$ & \textbf{474} &         $-$ \\
		\bottomrule
	\end{tabu}
	\label{tbl:further_experiments}
\end{table}

We provide a comparison of further solution methods in Table~\ref{tbl:further_experiments}.
$\texttt{SCC}^M$ denotes the SCC decomposition approach of Algorithm~\ref{alg:scc_eval} paired with the MEC decomposition as in Algorithm~\ref{alg:mec_dec_si}.
\texttt{LP} describes the LP-based mean-payoff solver MultiGain~\cite{DBLP:conf/tacas/BrazdilCFK15} (version 1.0.2) paired with the commercial LP solver Gurobi (version 7.0.1).
The table shows that for reasonably sized models, both of the additionally presented methods are significantly outperformed by our $\texttt{SCC}$ method.

\begin{table}[h]
	\centering
	\caption{Additional data gathered in the experiments.
		The first group shows the time taken for model construction and MEC decomposition, respectively.
		In the second group, some performance metrics of strategy iteration on these models are shown.}
	\setlength\tabcolsep{5pt}
	\begin{tabu} to \linewidth {l|cc|cccc}
		\textbf{Model} & Build & MEC & Steps & Gain & Bias & Changes \\
		\toprule
		%                          BUILD       MEC & s/g/b/i
		cs\_nfail3 (184, 38)    &      0 &       0 &     1 &    0 &    0 &           0 \\
		cs\_nfail4 (960, 176)   &      0 &       0 &     1 &    0 &    0 &           0 \\
		\midrule
		virus (809, 1)          &      0 &       0 &     1 &    0 &    0 &           0 \\
		\midrule
		phil\_nofair3 (856, 1)  &      0 &       0 &  13/6 &  5/0 &  7/5 &   1785/1104 \\
		phil\_nofair4 (9440, 1) &      1 &       0 & 18/15 &  7/5 & 10/9 & 27461/18936 \\
		\midrule
		sensors1 (462, 132)     &      0 &       0 &     7 &    0 &    6 &         537 \\
		sensors2 (7860, 4001)   &      1 &       0 &    14 &    6 &    7 &        9289 \\
		sensors3 (77766, 46621) &      3 &       1 &    16 &    8 &    7 &       84721 \\
		\midrule
		mer3 (15622, 9451)      &      1 &       0 &    13 &    7 &    5 &        8188 \\
		mer4 (119305, 71952)    &      5 &       1 &    15 &    7 &    7 &       73264 \\
		mer5 (841300, 498175)   &     41 &      11 &    15 &    7 &    7 &      567899 \\
		\bottomrule
	\end{tabu}
	\label{tbl:further_data}
\end{table}

Further data of the experiments is provided in Table~\ref{tbl:further_data}.
The first two columns show the time taken for model construction and MEC decomposition in seconds, respectively.
For all models (except \textbf{mer5}), the time taken is negligible.

The following four columns describe, from left to right, the number of strategy evaluations, the number of gain and bias improvements, and finally the total number of changes to the strategy.
For the \textbf{phil\_nofair} models, we included the values for both of the checked properties separately.

We want to highlight the number of improvement steps performed by the algorithm, which in the worst case is exponential in the number of states.
Nevertheless, it is small for all performed experiments and moreover does not significantly increase for larger models of the same type.

Note that for \texttt{SI}, \texttt{BSCC}, and \texttt{SCC} the presented numbers are equal, since we did not change the underlying principles of the strategy iteration algorithm.
We observed an exception to this for the \textbf{phil\_nofair} models.
There, \texttt{BSCC} has a slightly different number of strategy changes compared to \texttt{SCC}.
We suspect that this is due to small rounding errors.

\end{document}